 \documentclass[
 preprint,review,12pt]{elsarticle}

\usepackage{amssymb}
 \usepackage{amsthm}

\newtheorem{theorem}{Theorem}[section]

\newtheorem{lemma}[theorem]{Lemma}

\newtheorem{remark}[theorem]{Remark}

\journal{Statistics and Probability Letters}

\begin{document}

\begin{frontmatter}



\title{Directly Specifying the Power Semicircle Distribution}
\author{Hazhir Homei}


\address{Department of Statistics, Faculty of Mathematical
Sciences, University of Tabriz, P.O.Box 5166617766, Iran. Email:
homei@tabrizu.ac.ir, Fax: 0098 411 334 2102.}

\begin{abstract}
A new proof for a newly proved conjecture of Soltani and Roozegar
(2012) is provided; our proof does not make any use of the Stieltjes
transform unlike the proof of  Roozegar and Soltani (2013), 
and the distribution of power semicircle has been directly specified, 
contrary to the authors' claim in (Roozegar and Soltani 2013).
\end{abstract}

\begin{keyword}
Moments\sep Power Semicircle Distribution \sep Arcsin Distribution\sep Randomly Weighted Average. 
\MSC  60E05 \sep 62E15. 
\end{keyword}

\end{frontmatter}



\section{Introduction}
To set the stage, we briefly review the work of Soltani and Roozegar
(2012), who considered a randomly weighted average (RWA) of
independent random variables $X_1,\cdots,X_n$ defined by
$$ S_n = R_{1}X_1 + R_2X_2 + \cdots
+R_{n-1}X_{n-1} + R_nX_n,\;\;\;\; n\geq 2\;,\eqno(1)$$ where
random proportions are $R_{i}=U_{(i)}-U_{(i-1)}$, for $i=1,\cdots,n-1$, and $R_n= 1-\sum_{i=1}^{n-1} R_i$; and
$U_{(1)},\cdots,U_{(n-1)}$ are order statistics of a random sample
$U_1,\cdots,U_n$ from a uniform
distribution on [0,1], also $U_{(0)}=0$ and $U_{(n)}=1$.
If the random variables $X_1,\cdots,X_n$ are independent and have
common Arcsin distribution on $(-a,a)$, then $S_n$ will have a power
semicircle distribution on $(-a,a)$ with $\lambda=\frac{(n-1)}{2}$, i.e.
$$f(x;\lambda,a)=\frac{1}{\sqrt \pi {a^{2\lambda}}}\frac{\Gamma(\lambda+1)}{\Gamma(\lambda+\frac{1}{2})}(a^2-x^2)^{\lambda-\frac{1}{2}},\;\;\;|x|<a, a>0.$$
The authors established the cases of $n=2,3,4$ in Example 1 of Soltani and
Roozegar (2012), and later proved it for all $n$'s in Roozegar and Soltani (2013). In this short note, we will give a new proof for this identity, for every $n$, by using moments, and without making use of the Stieltjes transforms. 

\section{The Proof}
In order to prove the identity, we need the following lemma.

\begin{lemma}
For all positive integers $r \in \mathbb N$, we
have
$$\sum_{i_1+\cdots+i_n=r}^{}{r\choose
i_{1}, i_{2}, \cdots, i_{n}}
\frac{\Gamma(a_1+i_1)}{\Gamma(a_1)}\cdots\frac{\Gamma(a_n+i_n)}{\Gamma(a_n)}=\frac{\Gamma(\sum_{j=1}^{n}a_j+r)}
{\Gamma(\sum_{j=1}^{n}a_j)}.$$
\end{lemma}

\begin{proof}
Let the distribution of $f(x|p)$ be multinomial with the
parameters $(p_1,\cdots,p_n)$, and assume that $(p_1,\cdots,p_n)$ has dirichlet distribution with the parameters $(a_1,\cdots,a_n)$. So, the
distribution of $f(x)$ can be calculated, and the lemma is proved
considering the fact that the sum of $f(x)$ on its
support equals to one.
\end{proof}

\begin{theorem}
Assume that $a=1$ and that the random
variables $X_1,\cdots,X_n$ are independent and have common Arcsin
distribution on (-1,1). Then $S_n$ will have a power semicircle distribution on
(-1,1) with $\lambda=\frac{n-1}{2}$, i.e.
$$f(x;\lambda)=\frac{1}{\sqrt \pi}\frac{\Gamma(\lambda+1)}{\Gamma(\lambda+\frac{1}{2})}(1-x^2)^{\lambda-\frac{1}{2}}\;\;\;|x|<1.$$
\end{theorem}
\begin{proof}
First, we find the $r^{\rm th}$ moment of $S_n$ as follows:
$$E({S_n}^r)=\sum_{i_1+\cdots+i_n=r}^{}\frac{r!}{{i_1}!\cdots {i_n}!}E({R_1}^{i_1}\cdots {R_n}^{i_n})E({X_1}^{i_1})\cdots E({X_n}^{i_n}).$$
By using the dirichlet distribution, we have
$$E({S_n}^r)=\sum_{i_1+\cdots+i_n=r}^{}\frac{r!}{{i_1}!\cdots {i_n}!}{(n-1)!\frac{\Gamma(i_1+1)\cdots \Gamma(i_n+1)}{\Gamma(r+n)}}E({X_1}^{i_1})\cdots E({X_n}^{i_n}).$$
It is well known that
$$E({X_j}^{i_j})=\frac{1}{2}\frac{(1+(-1)^{i_j})
\Gamma(\frac{1}{2}+\frac{i_j}{2})}{\sqrt{\pi}\Gamma(1+\frac{i_j}{2})}, \textrm{ for } j=1,\cdots,n.$$
So, $E({S_n}^r)=$
$$\sum_{i_1+\cdots+i_n=r}^{}\frac{r!}{{i_1}!\cdots {i_n}!}{(n-1)!
\frac{\Gamma(i_1+1)\cdots \Gamma(i_n+1)}{\Gamma(r+n)}}$$
$${\frac{1}{2}\frac{(1+(-1)^{i_1})\Gamma(\frac{1}{2}+\frac{i_1}{2})}{\sqrt{\pi}\Gamma(1+\frac{i_1}{2})}}\cdots
\frac{1}{2}\frac{(1+(-1)^{i_n})\Gamma(\frac{1}{2}+\frac{i_n}{2})}{\sqrt{\pi}\Gamma(1+\frac{i_n}{2})}.$$
Since Arcsin distribution is symmetric about zero, the $r^{\rm th}$
moment is zero for odd r. Now we note that for even $r(=2k)$ if
$i_1+\cdots+i_n =r = 2k$ and one of ${i_{j}}'{s}$ is odd then
$1+(-1)^{i_j}=0$ so the corresponding summand will equal to zero.
Hence, we assume all ${i_{j}}'{s}$ to be even, so we write $2i_j$ in
place of $i_j$. Thus,
$$E({S_n}^{2k})=\sum_{i_1+\cdots+i_n=k}^{}\frac{(2k)!}{(2{i_1})!\cdots (2{i_n})!}(n-1)!$$
$$\frac{\Gamma(2{i_1}+1)\cdots\Gamma(2{i_n}+1)}{\Gamma(2k+n)}
\frac{1}{2}\frac{2\Gamma(\frac{1}{2}+i_1)}{\Gamma(\frac{1}{2})\Gamma(1+i_1)}\cdots\frac{1}{2}
\frac{2\Gamma(\frac{1}{2}+i_n)}{\Gamma(\frac{1}{2})\Gamma(1+i_n)}$$
$$=\frac{(2k)!(n-1)!}{\Gamma(2k+n)}\sum_{i_1+\cdots+i_n=r}^{}\frac{1}{{i_1}!\cdots{i_n}!}
\frac{\Gamma(\frac{1}{2}+i_1)}{\Gamma(\frac{1}{2})}\cdots
\frac{\Gamma(\frac{1}{2}+i_n)}{\Gamma(\frac{1}{2})}$$
$$=\frac{(2k)!(n-1)!}{\Gamma(2k+n)k!}\sum_{i_1+\cdots+i_n=r}^{}\frac{k!}{{i_1}!\cdots{i_n}!}
\frac{\Gamma(\frac{1}{2}+i_1)}{\Gamma(\frac{1}{2})}\cdots\frac{\Gamma(\frac{1}{2}+i_n)}{\Gamma(\frac{1}{2})}.$$
By using Lemma 2.1 we find that
$$E({S_n}^{2k})=\frac{(2k)!(n-1)!}{\Gamma(2k+n)k!}\frac{\Gamma(\frac{n}{2}+k)}{\Gamma(\frac{n}{2})}=
\frac{1}{\Gamma(\frac{1}{2})}\frac{\Gamma(\frac{n}{2}+\frac{1}{2})\Gamma(k+\frac{1}{2})}
{\Gamma(k+\frac{n}{2}+\frac{1}{2})}$$
So,
$$E({S_n}^r)=\left\{
\begin{array}{cl}
  0  & \textrm{ if }\ r=2k+1, \\
  \frac{\Gamma({k}+\frac{1}{2})\Gamma(\frac{n}{2}+\frac{1}{2})}
  {\sqrt{\pi}\Gamma({k}+\frac{n}{2}+\frac{1}{2})} & \textrm{ if }\ r=2k,
\end{array}
\right.
$$ 
is the $r^{\rm th}$ moment of the power semicircle distribution, and this  proves the theorem.
\end{proof}

\begin{remark}{\rm The restriction $a=1$ in Theorem 2.2. was just for
simplicity. In fact, a very similar argument can prove the theorem
for arbitrary $a$'s.}
\end{remark}

\end{document}